\documentclass[11pt]{article}

\usepackage[utf8]{inputenc}
\usepackage[T1]{fontenc}
\usepackage[english]{babel}

\usepackage{amsmath,amssymb,amsfonts,amsthm}
\usepackage{mathabx}

\usepackage{cleveref}

\usepackage{fullpage}

\newtheorem{lemma}{Lemma}
\newtheorem{theorem}{Theorem}

\newtheorem{corollary}{Corollary}
\theoremstyle{definition}
\newtheorem{definition}{Definition}
\theoremstyle{proof}

\crefname{lemma}{Lemma}{Lemmas}
\Crefname{lemma}{Lemma}{Lemmas}
\crefname{theorem}{Theorem}{Theorems}
\Crefname{theorem}{Theorem}{Theorems}
\crefname{corollary}{Corollary}{Corollaries}
\Crefname{corollary}{Corollary}{Corollaries}
\crefname{observation}{Observation}{Observations}
\Crefname{observation}{Observation}{Observations}
\crefname{definition}{Definition}{Definitions}
\Crefname{definition}{Definition}{Definitions}
\crefname{section}{Section}{Sections}
\Crefname{section}{Section}{Sections}
\crefname{figure}{Figure}{Figures}
\Crefname{figure}{Figure}{Figures}

\newcommand{\calT}{\mathcal T}
\newcommand{\calA}{\mathcal A}
\newcommand{\calF}{\mathcal F}
\newcommand{\calQ}{\mathcal Q}
\newcommand{\ceil}[1]{\left\lceil{#1}\right\rceil}
\newcommand{\floor}[1]{\left\lfloor{#1}\right\rfloor}

\newcommand{\eps}{\varepsilon}

\newcommand{\sm}{\,\setminus\,}
\newcommand{\abs}[1]{\left | #1 \right |}
\newcommand{\set}[1]{\left \{ #1 \right \}}

\usepackage{enumitem}

\usepackage[pdftex]{graphicx}
\usepackage[usenames,dvipsnames,table]{xcolor}
\definecolor{shade}{RGB}{235,235,235}

\newcommand{\Oh}{O}

\newcommand{\etal}{{\em et al.}}

\usepackage{booktabs}

\usepackage{todonotes}

\newcommand*\samethanks[1][\value{footnote}]{\footnotemark[#1]}
\usepackage{authblk}

\usepackage[ruled,linesnumbered]{algorithm2e}

\title{Optimal induced universal graphs and adjacency labeling for trees}
\author{Stephen Alstrup\thanks{Research partly supported by the FNU
    project AlgoDisc - Discrete Mathematics, Algorithms, and Data Structures.}}
\author{Søren Dahlgaard\thanks{Research partly supported by Mikkel Thorup's
    Advanced Grant from the Danish Council for Independent Research
    under the Sapere Aude research career programme.}}
\author{Mathias Bæk Tejs Knudsen\samethanks[1]\ \samethanks[2]}
\affil{University of Copenhagen,\\
    \tt{\{s.alstrup,soerend,knudsen\}@di.ku.dk}
}
\date{}

\begin{document}
\setcounter{page}{0}
\maketitle
\begin{abstract}
    We show that there exists a graph $G$ with $\Oh(n)$ nodes, such that
    any forest of $n$ nodes is a node-induced subgraph of $G$. Furthermore, for
    constant arboricity $k$, the result implies the existence of a graph with
    $\Oh(n^k)$ nodes that contains all $n$-node graphs of arboricity $k$
    as node-induced subgraphs, matching a $\Omega(n^k)$ lower bound. The lower
    bound and previously best upper bounds were presented in Alstrup and Rauhe
    [FOCS'02]. Our upper bounds are obtained through a $\log_2 n +\Oh(1)$
    labeling scheme for adjacency queries in forests.

    We hereby solve an open problem being raised repeatedly over decades, e.g.
    in Kannan, Naor, Rudich [STOC'88], Chung [J. of Graph Theory'90],
    Fraigniaud and Korman [SODA'10].
\end{abstract}

\thispagestyle{empty}

\newpage
\setcounter{page}{1}
\section{Introduction}
An \emph{adjacency labeling scheme} for a given family of graphs assigns
\emph{labels} to the vertices of each graph from the family such that given the
labels of two vertices from a graph, and no other information, it is possible
to determine whether or not the vertices are adjacent in the graph. The labels
are assumed to be bit strings, and the goal is to minimize the maximum label
size. A $k$-bit labeling scheme (sometimes denoted
$k$ labeling scheme) uses at most $k$ bits per label. In information theory
adjacency labeling schemes studies goes back to the
1960's~\cite{Breuer66,BF67}, and efficient labeling schemes were introduced
in~\cite{KNR92,muller}. Adjacency labeling schemes are also called
\emph{implicit representation of
graphs}~\cite{spinrad2003efficient,wiki:implicit}.

As an example let $\calA_n$ denote the family of forests with $n$ nodes. Given
a forest $F\in \calA_n$, do the following: Root the trees of $F$ and assign
each node with an id from $[0,n-1]$. Let the label of each node be its id
appended with the id of its parent. A test for adjacency is then simply
to test whether the id of one of the nodes equals the stored
parent id of the other node. The labels assigned to the nodes have length $2
\lceil \log n \rceil$ bits\footnote{Throughout this paper we use $\log$ for
$\log_2$.}.

Closely related to adjacency labeling schemes are \emph{induced-universal
graphs} also studied in the 1960's~\cite{moon1965minimal,Rado64}. A graph
$G=(V, E)$ is said to be an induced-universal graph for a family $\cal F$ of
graphs, if it contains all graphs in $\cal F$, as node-induced subgraphs. A
graph $H=(V',E')$ is contained in $G$ as a node-induced subgraph if $V'
\subseteq V$ and $E'=\{(v,w)|v,w \in V' \wedge (v,w) \in E\}$. We define
$g_v(\calF)$ to be the smallest number of nodes in any
induced-universal graph for $\calF$. From~\cite{KNR92} (some details given
in~\cite{alstruprauhe,spinrad2003efficient}) we have:

\begin{theorem}[\cite{KNR92}] \label{KNRreduction}
	A family, $\calF$, of graphs has a $k$-bit adjacency labeling scheme with
	unique labels iff $g_v(\calF) \leq 2^k$.
\end{theorem}

Labels being unique means that no two nodes in the same graph from $\calF$ will be given the same label.

Combining the $2\lceil \log n \rceil$-bit labeling scheme above with
Theorem~\ref{KNRreduction} gives $g_v(\calA_n)=\Oh(n^2)$. Closely related, a \emph{universal} graph for ${\cal F}$ is a graph that contains each graph from ${\cal F}$ as a subgraph, not necessarily induced. The challenge is to construct universal graphs with as few edges as possible. Let $f_e(\calF)$ denote the minimum
number of edges in a universal graph for $\calF$. In a series of
papers~\cite{BCEGS82,Chung90,CG78,CG79,CG83,CGP76,Nebesky75} it was
established that $f_e(\calA_n)=\Theta (n\log n)$.
Let $G=(V,E)$ be any universal graph for any family $\mathcal{H}$ of acyclic
graphs. In~\cite{Chung90} Chung shows $g_v(\mathcal{H}) \leq 2|E|+|V|$ and,
combined with bounds for $f_e(\calA_n)$, concludes that $g_v(\calA_n)= \Oh(n
\log n)$. As the bounds for $f_e(\calA_n)$ are tight it is not possible to
improve the bounds for $g_v(\calA_n)$ using the techniques of \cite{Chung90}.
However, for the family of graphs of forests with bounded degree and $n$ nodes,
denoted $\mathcal{A}^B_n$, there exists a universal graph with $n$ nodes and
$\Oh(n)$ edges~\cite{BCLR86,BCLR89}, giving
$g_v(\mathcal{A}^B_n)=\Oh(n)$~\cite{Chung90}.

Chung's results~\cite{Chung90} combined with Theorem~\ref{KNRreduction} give a
$\log n+\log \log n+\Oh(1)$
adjacency labeling scheme for forests, and $\log
n+\Oh(1)$ for bounded degree forests. In 2002 Alstrup and
Rauhe~\cite{alstruprauhe} gave a $\log n + \Oh(\log^* n)$ adjacency labeling
scheme for general forests\footnote{$\log^*$ is the number of times $\log$
should be iterated to get a constant.}. Adjacency labeling schemes using $\log
n + \Oh(1)$ bits are given in~\cite{bonichon2006short,cyrilbounded,bonichon}
for bounded degree forests and caterpillars, in~\cite{fraigniaudkorman2} for
bounded depth trees, and in~\cite{Fraigniaud2009randomized} the case allowing
1-sided errors. Adjacency labeling schemes for forests are also considered
in~\cite{icalpnoy14,KM01}. Table~\ref{tab:adjacency2} summarizes the results.

\begin{table}[htbp]
    \centering
    \begin{tabular}{l|c|c}
        \toprule
        \bf Graph family &  \bf Upper bound & \bf Reference \\
        \midrule
        Forests of bounded degree & $\Oh(n)$& \cite{Chung90} \\
        Forests  & $n2^{\Oh(\log^*n)}$& \cite{alstruprauhe} \\
        Caterpillars & $\Oh(n)$& \cite{bonichon2006short}  \\
        Trees of depth $d$  & $\Oh(nd^3)$ & \cite{fraigniaudkorman2} \\
        \midrule
        Forests & $\Oh(n)$ & This paper \\
        \bottomrule
    \end{tabular}
    \caption{Size of induced-universal graphs for various families of
    forests.}
    \label{tab:adjacency2}
\end{table}

While minimizing the label size is the main goal of a labeling scheme, we
sometimes also seek to reduce the running time. The time used to assign labels
to the nodes is called the \emph{encoding time}, and the time used to decide
whether two nodes are adjacent or not is called the decoding time.
In~\cite{bonichon2006short,cyrilbounded,bonichon} described above the encoding
time is $\Oh(n)$ and decoding time is $\Oh(1)$.

Addressing a problem repeatedly raised the last decades, e.g.
in~\cite{icalpnoy14,bonichon2006short,Chung90,CG78,CG79,Fraigniaud2009randomized,fraigniaudkorman2,gavoille2007shorter,KNR92}
we show:

\begin{theorem} \label{optimaladjacency}
    There exists an adjacency labeling scheme for $\calA_n$
    using unique labels of length $\log n+\Oh(1)$ bits with $\Oh(1)$ decoding
    time and $\Oh(n)$ encoding time in the word-RAM model.
\end{theorem}

In our solution the decoder does not know $n$ in advance. The importance of the
problem is emphasized by it repeatedly and explicitly being raised as a central
open problem (see appendix~\ref{whatisthat}). Theorem~\ref{optimaladjacency}
establishes that adjacency labeling in forests requires $\log n +
\Theta(1)$ bits. To see this, consider the path of length $n$ as well as the
star on $n$ nodes. These two graphs may share at most $n/2$ labels,
giving a $\log 1.5n = \log n + \Omega(1)$ lower bounds. We note that
this lower bound may be slightly improved using the result of
\cite{Otter48}.


\subsection{Graphs with bounded arboricity}
Let $\calF$ and $\calQ$ be two families of graphs and let $G$ be an
induced-universal graph for $\calF$.
Suppose that every graph in the family $\mathcal{Q}$ can be edge-partitioned
into $k$ parts, each of which forms a graph in $\mathcal{F}$.
In this case, it was shown by Chung~\cite{Chung90} that
$g_v(\mathcal{Q}) \leq |V(G)|^k$. She considered the family, $\calA_n^k$ of
graphs with arboricity $k$ and $n$ nodes.
A graph has \emph{arboricity} $k$ if the edges of
the graph can be partitioned into at most $k$ forests. By combining the
above result with $g_v(\calA_n) = O(n\log n)$ she showed that
$g_v(\calA_n^k)=\Oh((n\log n)^k)$ improving the bound of $n^{k+1}$
from~\cite{KNR92}. For constant arboricity $k$, it follows
from~\cite{alstruprauhe} that $\Omega(n^k)= g_v(\calA_n) \leq n^k
2^{\Oh(\log^*n)}$.  Combining Chung's reduction~\cite{Chung90}
with Theorem~\ref{KNRreduction} and~\ref{optimaladjacency} we show that:

\begin{theorem} \label{arboricity}
There exists an induced-universal graph of size $\Oh(n^k)$ for the family of
graphs with constant arboricity $k$ and $n$ nodes.
\end{theorem}

Achieving results for bounded degree graphs by reduction to bounded arboricity
graphs is e.g. used in~\cite{KNR92}. This can be done as graphs with bounded
degree $d$ have arboricity bounded by
$\floor{\frac{d}{2}}+1$~\cite{chartrand68,Lovasz66}.

\subsection{Adjacency labeling and induced-universal graphs for other families}
Induced-universal graphs (and hence adjacency labeling schemes) are given for
tournaments~\cite{BW75,moon1968topics}, hereditary
graphs~\cite{Lozin97,Scheinerman199416}, threshold graphs~\cite{CPC:1772872},
special commutator graphs~\cite{Pisanski:1989:UCG:82296.82313}, bipartite
graphs~\cite{Lozin2007}, bounded degree graphs~\cite{TV98}, and other
cases~\cite{bollobas1981graphs,JGT:JGT3190120204}.  Using  universal graphs
constructed by Babai \etal~\cite{BCEGS82}, Bhatt \etal~\cite{BCLR89} and Chung
\etal~\cite{CG78,CG79,CG83,CGP76}, Chung~\cite{Chung90} obtains the current
best bounds for e.g. induced-universal graphs for bounded degree graphs being
planar or outerplanar. Many other results use reductions from~\cite{Chung90},
e.g. the induced-universal graphs for bounded degree graphs
\cite{Butler_induced-universalgraphs,Esperet2008}. The result from
~\cite{Esperet2008}, as many others, is achieved by reduction to a universal
graph with bounded degree~\cite{alon2007sparse,AlonCapalbo2008}. Other results
for universal graphs is e.g. for families of graphs such as
cycles~\cite{Bondy71}, forests~\cite{CGC81,Fishburn85}, bounded degree
forests~\cite{BCLR86,FP87}, and graphs with bounded path-width~\cite{TUK95}.
In~\cite{AlstrupKTZ14} they give a $(\ceil{n/2}+4)$-bit adjacency labeling
scheme for general undirected graphs, improving the $(\lfloor
n/2\rfloor+\lceil\log n\rceil)$ bound of~\cite{moon1965minimal}, almost
matching an $(n-1)/2$ lower bound~\cite{KNR92,moon1965minimal}. An overview of
induced-universal graphs and adjacency labeling can be found
in~\cite{AlstrupKTZ14}.

\subsection{Second order terms for labeling schemes are theoretically significant}
Above it is shown that for adjacency labeling significant work has been done
optimizing the second order term. This is also true for other labeling scheme
operations. E.g. the second order term in the ancestor relationship is improved
in a sequence of STOC/SODA
papers~\cite{AKM01,AlstrupBR03,AR02,fraigniaudkorman2,fraigniaudkorman}
(and~\cite{abiteboul,KMS02})  to $\Theta(\log \log n)$, giving labels of size
$\log n+\Theta(\log \log n)$. Lastly, an algorithm giving both a simple and
optimal scheme was given in \cite{Dahlgaard15ancestry}. Somewhat related,
\emph{succinct data structures} (see,
e.g.,~\cite{DPT10,FarzanM13,FarzanM14,MunroRRR12,MR97,patrascu08succinct})
focus on the space used in addition to the information theoretic lower bound,
which is often a lower order term with respect to the overall space used.

\subsection{Labeling schemes in various settings and applications}
By using labeling schemes, it is possible to avoid costly access to large
global tables, computing instead locally and distributed. Such properties are
used in applications such as XML search engines~\cite{AKM01}, network routing
and distributed algorithms~\cite{Cowen01,EilamGP03,Gavoille01,ThZw05}, dynamic
and parallel settings ~\cite{CohenKaplan2010,dynamicKormanP07}, and various
other applications~\cite{Korman2010,peleg2,SK85}.

Various computability requirements are sometimes imposed on  labeling
schemes \cite{AKM01,KNR92,siamcompKatzKKP04}. This paper assumes the RAM model
and mentions the time needed for encoding and decoding in addition to the label size.


Closely related to adjacency is small distances in trees. This is studied
by Alstrup et al. in~\cite{alstrupbillerauhe} who among
other things give a $\log n +\Theta(\log \log n)$ labeling scheme supporting
both parent and sibling queries. General distance labeling schemes for various
families of graphs exist, e.g., for trees~\cite{alstrupbillerauhe,peleg},
bounded tree-width, planar and bounded degree graphs~\cite{Gavoille200485},
some non-positively curved plane~\cite{CDV06}, interval~\cite{GP08} and
permutation graphs~\cite{BG09}, and general graphs~\cite{grahampollak,winkler}.
In~\cite{Gavoille200485} it is proved that distance labels require
$\Theta(\log^2 n)$ bits for trees. Approximate distance labeling schemes are
also well studied; see
e.g.,~\cite{GuptaKL03,gupta2005traveling,KL06,Talwar04,Thorup2004distance,TZ01a}.
An overview of distance labeling schemes can be found in~\cite{distalstrup},
and a more general labeling survey can be found in an overview
in~\cite{gavoillepeleg}.

\section{Preliminaries}
In this section we introduce some well-known results and notation.
Throughout this paper we use the convention that $\lg x = \max(1,
\log_2 x)$ for convenience. We assume the word-RAM model of computation.

\paragraph{Trees}
Let $\calT_n$ denote the family of all rooted trees of size $n$ and let $T\in
\calT_n$. We denote the nodes of $T$ by $V(T)$
and the edges by $E(T)$. We let $\abs{T}$ denote the number of nodes in $T$.
For a node $u\in V(T)$, we let $T_u$ denote the subtree
of $T$ rooted in $u$. A node $u$ is an \emph{ancestor} of a node $v$ iff it is
on the unique path from $v$ to the root. In this case we also say that $v$ is a
\emph{descendant} of $u$. A \emph{caterpillar} is a tree whose non-leaf nodes
induce a path. Throughout the paper we will only consider
adjacency labeling in trees, as we may add an ``imaginary root'' to any
forest on $n$ nodes turning it into a tree of size $n+1$. To do this we expend
at most one extra bit to distinguish this from actual nodes.


\paragraph{Heavy-light}
For a node $u$ with children $children(u)=v_1,\ldots, v_k$, with $|T_{v_k}|\geq |T_{v_i}|$
for all $i<k$, we say that the edge $(u, v_k)$ is \emph{heavy}, and the remaining edges
$(u,v_i)$ are \emph{light}. We say that
$heavy(u) = v_k$ is the \emph{heavy child} of $u$. A node $u$ for
which the edge $(parent(u),u)$ is light is called an \emph{apex node}.
For convenience we also define the root to be an apex node. For a node $u$, we
define $children(u)\sm \{heavy(u)\}$ to be the \emph{light children} of $u$.
This is called a \emph{heavy-light decomposition}~\cite{sleatortarjan} as it
decomposes the tree into paths of heavy edges (\emph{heavy paths}) connected by
light edges. We define the \emph{light subtree} of a node $u$ to be $T_u^\ell =
T_u\setminus T_{heavy(u)}$. For a leaf $u$, $T_u^\ell = T_u = u$. The \emph{light
depth} of a node $u$ is the number of light
edges on the path from $u$ to the root. The \emph{light height} of a node $u$
is the maximum number of light edges on a path from $u$ to a leaf in $T_u$.

\begin{lemma}~\cite{sleatortarjan}
    \label{lightHeightLemma}
    Given a tree $T$ and $u \in V(T)$ with light height $x$, $|T_u| \geq 2^{x+1}-1$.
\end{lemma}

\paragraph{Bit strings}
A bit string $s$ is a member of the set $\{0,1\}^*$. We denote the length of a
bit string $s$ by $|s|$, the $i$th bit of $s$ by $s_i$, and the concatenation of
two bit strings $s,s'$ by $s\circ s'$ (i.e. $s = s_1\circ s_2\circ\ldots\circ
s_{|s|}$). We say that $s_1$ is the most significant bit of $s$ and $s_{|s|}$
is the least significant bit. For an integer $x$ we let $0^x$ and $1^x$ denote
the strings consisting of exactly $x$ $0$s and $1$s respectively. Let $a$ be an
integer and let $s$ be the bit string representation of $a$.
Define the function $wlsb(a, k)$ to be $s_1\circ s_2\circ\cdots\circ
s_{|s|-k}$, i.e. the bit string of $a$ without the $k$ least significant
bits. When $k>|s|$ we define $wlsb(a, k)$ to be the empty string.
When constructing a labeling scheme we often wish to concatenate several
bit strings of unknown length. We may do this using the Elias $\gamma$ code
\cite{Elias75} to encode a length $k$ bit string with $2k$ bits and decode
it in $O(1)$ time for $k=O(w)$\footnote{Here, $w$ is the word size.}, using
standard bit operations.

For an integer $a$ we will often use $a$ to denote the bit string
representation of $a$ when it is clear from the context. We will use
$[a]_\gamma$ to denote the Elias $\gamma$ encoding of $a$.

\paragraph{Labeling schemes}
An \emph{adjacency labeling scheme} for trees of size $n$ consists of an
\emph{encoder}, $e$, and a \emph{decoder}, $d$. Given a tree $T\in\calT_n$, the
encoder computes a mapping $e_T : V(T)\to \set{0,1}^*$ assigning a \emph{label} to
each node $u\in V(T)$. The decoder is a mapping $d:\set{0,1}^*\times \set{0,1}^*\to
\{{\tt True},{\tt False}\}$ such that given any tree $T\in\calT_n$ and any
pair of nodes $u,v\in V(T)$ we have $d(e_T(u),e_T(v)) = {\tt True}$ iff $(u,v)\in
E(T)$. Note that the decoder does not know $T$. The \emph{size} of a labeling scheme is defined as the maximum label size
$|e_T(u)|$ over all trees $T\in\calT_n$ and all nodes $u\in V(T)$. If for all trees $T\in\calT_n$ the mapping $e_T$ is injective we say that the
labeling scheme assigns \emph{unique} labels. The labeling schemes constructed
in this paper all assign unique labels and the decoder does not know $n$.

\paragraph{Approximation}
Given a non-negative integer $a$ and a real number $\eps > 0$, a $(1+\eps)$-approximation of $a$
is an integer $b$ such that $a\le b < (1+\eps)a$. We also define
$b = 0$ to be the unique $(1+\eps)$-approximation of $a = 0$.
\begin{lemma}\label{lem:approximation}
    Given an integer $a$ and a number $\eps\in (0,1]$, we can find a
    $(1+\eps)$-approximation and represent it using $O(\lg
    \lg a + \lg \frac{1}{\eps})$ bits.
    Furthermore, if $\eps = \frac{1}{\delta}$, where $\delta$ is a
    positive integer that can be stored using $O(1)$ words, we can find this
    approximation in $O(1)$ time.
\end{lemma}
\begin{proof}
We will use a single bit to distinguish between the cases $a = 0$ and $a > 0$,
so assume $a > 0$. Let $\delta = \ceil{\eps^{-1}}$ and $\eps' = \delta^{-1}$.
Let $k = \ceil{\log_{1+\eps'} a}$. Then $(1+\eps')^k \ge a > (1+\eps')^{k-1}$.
Hence if we let $b = (1+\eps')^k$ we have $a \le b < a(1+\eps') \le a(1+\eps)$.
In order to encode $b$ it suffices to encode $\delta$ and $k$. We can do this
using $2\ceil{\lg \delta} + 2\ceil{\lg k}$ bits using the Elias $\gamma$
coding. Note that:
\[
    k-1 < \log_{1+\eps'} a = \frac{\log_2 a}{\log_2 (1 + \eps')}
\]
Taking $\log_2$ gives:
\begin{align*}
    \log_2 (k-1)
    &< \log_2 \log_2 a - \log_2 \log_2 (1+\eps') \\
    &= \log_2 \log_2 a + O \left ( 1 + \log_2 \frac{1}{\eps'} \right ) \\
    &= \log_2 \log_2 a + O \left ( 1 + \log_2 \frac{1}{\eps} \right )
\end{align*}
Hence $\lg k = O\left ( \lg \lg a + \lg \frac{1}{\eps} \right )$, and since
$\lg \delta \le 1 + \lg \frac{1}{\eps}$
the proof is finished.
\end{proof}

We will use $\FuncSty{Approx}(a,\eps)$ to denote a function returning a
$(1+\eps)$-approximation of $a$ as described above.

\section{A simple scheme for caterpillars}\label{sec:caterpillar}
As a warmup, we describe a simple adjacency labeling scheme of size $\lg n +
O(1)$ for caterpillars. The idea is to use a variant of this scheme recursively
when labeling general trees. The scheme we present uses ideas similar to that of~\cite{bonichon2006short}.

Let $p = (u_1,\ldots,u_{|p|})$ be a longest path of the caterpillar and root the
tree in $u_1$. We assign an id and an interval
$I(u_i) = [id(u_i), id(u_i)+l(u_i))$ to each node
$u_i$, such that $id(v)\in I(u_i)$ iff $v$ is a non-root apex node (all leaves
except $u_{|p|}$ are apex nodes) and $u_i$ is
the parent of $v$. The ids of the $u_i$s are assigned such that given the label
of $u_i$ we can deduce $id(u_{i+1})$ for $i < |p|$. We first calculate the interval sizes $l$ and
next assign the $id$s. Both steps can be done in $O(n)$ time.

\paragraph{Interval sizes}
Let $\gamma_i = \ceil{\lg |T_{u_i}^\ell|}$. For each node $u_j$ now define the
$|p|$-dimensional vector $\beta_j$ as $\beta_j(i) = \gamma_j - |i-j|$. Let
$k_i = \max_{j=1\ldots |p|} \beta_j(i)$. This ensures that
$(k_i-k_{i+1})\in\{-1,0,1\}$ for all $i\in\{1,\ldots,|p|-1\}$. The
process is illustrated in Figure~\ref{fig:betas}. The interval size of node $u_i$ is
now set to $l(u_i)=2^{k_i}$.

\begin{figure}[htbp]
    \centering
    \includegraphics[width=.4\textwidth]{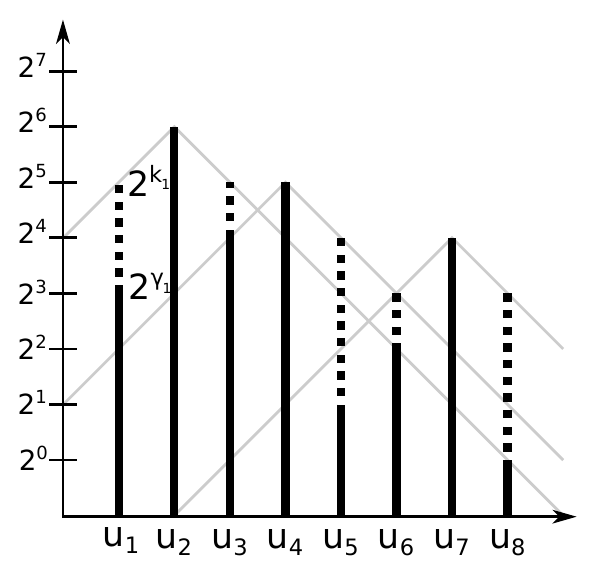}
    \caption{Example of how the $\beta_j$s are used to ensure that neighbouring
    nodes have $(k_i - k_{i+1})\in \{-1,0,1\}$.}
    \label{fig:betas}
\end{figure}

\paragraph{Id assignment}
The idea is to assign $id(u_i)$ such that the
$k_i$ least significant bits of $id(u_i)$ are all $0$. We first assign
the id for $u_1$ and its children, then $u_2$ and its children, etc. The
procedure is as follows:
\begin{enumerate}
    \item Assign $id(u_i) = x$, where $x$ is the smallest integer having $0$ as
        the $k_i$ least significant bits satisfying $x\ge id(u_{i-1}) +
        l(u_{i-1})$. For $u_1$ we set $id(u_1) = 0$.
    \item Let $v_1,\ldots,v_{|T_{u_i}^\ell|-1}$ be the light children of
        $u_i$. Assign $id(v_j) = id(u_i) + j$. Note that
        $id(v_{|T_{u_i}^\ell|-1}) < id(u_i) + l(u_i)$.
\end{enumerate}

\paragraph{The label}
For a node $u_i\in p$ we assign the label
\[
    \ell(u_i) = type(u_i)\circ [k_i]_\gamma\circ wlsb(id(u_i), k_i)\ ,
\]
and for $v\notin p$, assign the label
\[
    \ell(v) = type(v)\circ id(v)\ .
\]
Here $type(u)$ is $1$ if $u\notin p$. Otherwise, $type(u_i)$ is
$0xx$, where $xx$ is either $00$, $01$, $10$ or $11$ corresponding to the
following four cases: (00) $u_i=u_{|p|}$, (01) $k_i = k_{i+1} - 1$, (10) $k_i =
k_{i+1}$, and (11) $k_i = k_{i+1} + 1$.

\paragraph{Label size}
First, we let $N$ denote the maximum $id$ assigned by the encoder. Then the label size for a node $u_i\in p$ is $\le 3+2\ceil{\lg k_i}+\ceil{\lg N} - k_i$ and for $v\notin p$, it is $
\le 1 + \ceil{\lg N}$. We will now bound $N$:


\begin{lemma}\label{lem:caterN}
    Given a caterpillar $T$ with $n$ nodes, the maximum id assigned by our
    encoder, $N$, satisfies
    \[
        N\le 12n\ .
    \]
\end{lemma}
\begin{proof}
    First, observe that the number of ids skipped between
    $id(u_{i-1})+l(u_{i-1})$ and $id(u_i)$ is at most $2^{k_i}-1$ as any
    set of $2^{k_i}$ consecutive integers must contain at least one integer
    with $k_i$ $0s$ as least significant bits. Thus, the maximum id is
    bounded by $\sum_{i=1}^{|p|} \left(2^{k_i} -1+ l(u_i)\right) = 2\cdot
    \left(\sum_{i=1}^{|p|} 2^{k_i}\right) - |p|$ and we can bound this using
\[
\left(\sum_{i=1}^{|p|} 2^{k_i}\right)
        \le  \left(\sum_{i=1}^{|p|} \sum_{j=1}^{|p|}
    2^{\beta_j(i)}\right) =  \left(\sum_{j=1}^{|p|} \sum_{i=1}^{|p|}
    2^{\beta_j(i)}\right)
\le \left(\sum_{j=1}^{|p|} \sum_{i=-\infty}^{\infty}
    2^{\gamma_j - |i|}\right) =  \left(\sum_{j=1}^{|p|} 3\cdot 2^{\gamma_j}\right)
\]
concluding that $N \le 12n - |p|$
\end{proof}

\paragraph{Decoding}
Given the labels of $u,v\notin p$ we always answer {\tt False}.

Now assume that we are given the label of at least one node $u_i\in p$. First
we deduce $id(u_i)$ using $[k_i]_\gamma$ and $wlsb(id(u_i),k_i)$.
This also gives us $l(u_i) = 2^{k_i}$. Now there are two cases:
\begin{enumerate}
    \item If the other label is for a node $v\notin p$, we simply read
        $id(v)$ and answer {\tt True} if $id(v)\in [id(u_i), id(u_i) +
            l(u_i))$. Otherwise we answer {\tt False}.
    \item If the other label is for $u_j\in p$, assume without
        loss of generality that $id(u_j) > id(u_i)$. If $type(u_i) = 001$, set
        $x$ to be the smallest integer with the $k_i + 1$ least
        significantly bits set to $0$ satisfying $x\ge id(u_i) + l(u_i)$. If $x
        = id(u_j)$ answer {\tt True}, otherwise answer {\tt False}.

        The other types can be handled similarly.
\end{enumerate}

\section{An optimal scheme for general trees}

In this section we prove Theorem~\ref{optimaladjacency}. Similar to the
caterpillar scheme presented in the previous section we assign an id,
$id(u)$, and interval, $I(u)$, to each node. The interval and id of a node is
assigned such that $id(v)\in I(u)$ iff $v\in T_u^\ell$. The label of a node $u$
will be assigned such that we can infer the following information (loosely
speaking) directly from the label:
\begin{itemize}
    \item The id of the node $u$.
    \item The id of $u$'s heavy child, $heavy(u)$.
    \item The interval $I(u)$ containing the ids of all nodes in $u$'s light subtree.
    \item Auxilliary information to help decide whether $u$ is a light child of
        another node.
\end{itemize}
In order to store this information as part of the label, each node will be
assigned an id with a number of trailing zero bits proportional to the
logarithm of its interval
size corresponding to the $k_i$s of Section~\ref{sec:caterpillar}. Furthermore, we
ensure that the interval size for a node $u$ is proportional to $|T_u^\ell|$
(or simply $|T_u|$ for apex nodes), and call this the \emph{light weight} of
$u$ denoted by $lw(u)$. Intuitively this ensures that nodes with large subtrees
have more ``bits to spare''.

The labels are assigned using a similar two-step procedure as in
Section~\ref{sec:caterpillar}. In the first step we assign the light weight of each
node using a recursive procedure, and in the second step we assign the actual
ids of the nodes based on the given weights. Both steps are handled in $O(n)$
time.
In order to bound the maximum id assigned we introduce the notion of path
weights (to be defined later). The \emph{path weight} of a heavy path $p$ is
denoted $pw(u)$, where $u$ is the apex node of $p$.


\subsection{Weight classes and restricted light depth}
The auxilliary information mentioned above is primarily used to determine
adjacency between an apex node and its parent. A classic way of doing this is
to use the light depth of both nodes and check that it differs by exactly one.
However, the light depth of a node with a small subtree could potentially be
big in comparison, and thus we cannot afford to store it. To deal with this we
introduce the following notion of weight classes and restricted light depth:
\begin{definition}\label{def:wc}
    Let $T$ be a rooted tree and $u$ some node in $T$. Define
    \begin{equation}
        \gamma(u) =
        \begin{cases}
            \floor{\lg |T_u|} & \text{if $u$ is an apex node} \\
            \floor{\lg |T_u^\ell|} & \text{otherwise.}
        \end{cases}
    \end{equation}
    The \emph{weight class} of $u$ is defined as $wc(u) = \floor{\lg
    \gamma(u)}$.
\end{definition}
\begin{definition}\label{def:rld}
    Let $T$ be a rooted tree and $u$ some node in $T$. Define $wtop(u)$ to be
    the ancestor of $u$ with smallest depth such that every node on the path
    from $u$ to $wtop(u)$ has weight class $\le wc(u)$. The \emph{restricted
    light depth} of $u$ is the number of light edges on the path from $u$ to
    $wtop(u)$ and is denoted by $rld(u)$.
\end{definition}
An illustration of these definitions can be seen in Figure~\ref{fig:rld}.
\begin{figure}[htbp]
    \centering
    \includegraphics[width=.6\textwidth]{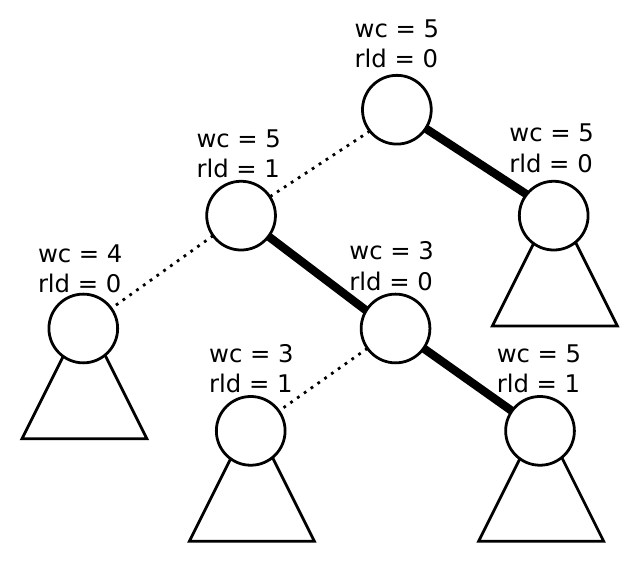}
    \caption{Example of weight classes and restricted light depths in a tree.
        The dotted and solid lines correspond to light and heavy edges
        respectively.}
    \label{fig:rld}
\end{figure}

When assigning the interval $I(u)$, we will split it into a sub-interval for
each weight class $i\le wc(u)$.

We will now show some properties related to weight classes and restricted light
depth. We will use the definitions of $\gamma(u)$ and $wtop(u)$ as described in
Definitions~\ref{def:wc} and \ref{def:rld}.

\begin{lemma}
    \label{rldBound}
    Let $u$ be any node, then $rld(u) \le 2\gamma(u)+1$.
\end{lemma}
\begin{proof}
    Let $v$ be the apex node on the path from $v$ to $wtop(u)$ with the smallest
	depth. (If no such node exist $rld(u) = 0$ and the result is trivial.)
    We note that $v$ must have light height $\ge rld(u)-1$, so by
    Lemma~\ref{lightHeightLemma} $\abs{T_v} \ge 2^{rld(u)}-1$ and therefore
    $\gamma(v) \ge rld(u)-1$. So
    \[
        2\gamma(u) \ge 2^{wc(u)+1} \ge 2^{wc(v)+1} \ge
        \gamma(v) \ge rld(u)-1
    \]
    which finishes the proof.
\end{proof}
\begin{lemma}
    \label{rldToApexAncestor}
    Let $u$ be an ancestor of $v$ such that $u$ is an apex node and $wc(u) = wc(v)$.
    Let $k$ be the number of light edges on the path from $u$ to $v$. Then
    $rld(v) = rld(u) + k$.
\end{lemma}
\begin{proof}
    Any node in $u$'s subtree must have weight class $\le wc(u)$
    since $u$ is an apex node. Since $wc(u) = wc(v)$ every node on the path
    from $v$ to $wtop(u)$ must have weight class $\le wc(v)$.
    Thus $wtop(v) = wtop(u)$ and there are $rld(u) + k$ light edges on the path from
    $v$ to $wtop(u)$, i.e. $rld(v) = rld(u) + k$.
\end{proof}
\begin{lemma}
    \label{rldToParent}
    Let $u$ be the parent of an apex node $v$. If $wc(v) < wc(u)$ then $rld(v) = 0$,
    and if $wc(v) = wc(u)$ then $rld(v) = rld(u)+1$.
\end{lemma}
\begin{proof}
    If $wc(v) < wc(u)$ then $v$ has restricted light depth $0$ so assume that
    $wc(u) = wc(v)$. Let $w$ be the apex node on $u$'s heavy path (possibly
    $u$ itself). Then first assume
	that $wc(w) = wc(u)$. By Lemma~\ref{rldToApexAncestor} $rld(u) = rld(w)$ and
	$rld(v) = rld(w)+1$ and the claim is true. Now assume that $wc(w) > wc(u)$.
	Then $rld(u) = 0$ and $rld(v) = 1$ and the claim is true as well. Since
	$wc(w) < wc(u)$ is impossible the proof is finished.
\end{proof}

\subsection{Weight assignment}
We will now see how to assign path weights and light weights to the nodes.
The idea is to consider an entire heavy path as a ``recursive caterpillar''
and use ideas similar to those of Section~\ref{sec:caterpillar}. Consider any heavy
path $p = (u_1,u_2,\ldots,u_{\abs{p}})$ in order where $u_1$ is the apex node.
For each $u \in p$ we do the following:
\begin{enumerate}
    \item For each light-child $v$ of $u$ we recursively calculate $pw(v)$.
    \item For every weight class $i \le wc(u)$, let $b_i$ be the sum of $pw(v)$
          for all light children $v$ of $u$ with weight class $wc(v) = i$.
    \item We use the convention that $a_0(u) = 0$, and for $i=1,\ldots,wc(u)$ we
          let $a_i(u)$ be a  $\left(1 +
          \frac{1}{(\gamma(u))^3}\right)$-approximation of $a_{i-1}(u) +
          b_i(u)$.
    \item We then define the light weight of $u$ as $lw(u) = 1 + a_{wc(u)}(u)$.
\end{enumerate}

For each $i = 1,2,\ldots,\abs{p}$ we let $k'(u_i) = \gamma(u_i) - \ceil{2 \lg \gamma(u_i)} + 1$.
We choose $k(u_1),\ldots,k(u_{\abs{p}})$ such that $k(u_i) \ge k'(u_i)$
for every $i = 1,\ldots,\abs{p}$ and $k(u_i) - k(u_{i+1}) \in \set{-1,0,1}$
for all $i = 1,\ldots,\abs{p}-1$. We do this in the same manner as in
Section~\ref{sec:caterpillar}
when we constructed the labeling scheme for the caterpillar, see Figure~\ref{fig:betas}.

The path weight of $u_1$ is defined as
$pw(u_1) = \sum_{i=1}^{\abs{p}} \left ( lw(u_i) + 2^{k(u_i)} - 1 \right )$. By
this definition, the path weight of a leaf apex node is $1$.

Pseudocode for the function $\FuncSty{Assign-Weight}$ is available in
Algorithm~\ref{alg:assign_weight}.

\begin{algorithm}
    \caption{\FuncSty{Assign-Weight}}
    \label{alg:assign_weight}
    \DontPrintSemicolon
    \SetKwInOut{Input}{input}\SetKwInOut{Output}{output}
    \Input{Heavy path $p = (u_1,\ldots, u_t)$ represented by $u_1$.}
    \Output{path weight of $p$.}
    \BlankLine
    \For{$i=1\to t$}{
        $a_0(u_i)\leftarrow 0$\;
        \For{$j = 1\to wc(u_i)$}{\label{line:for_weight}
            $b_j \leftarrow 0$\;
            \For{$v\in \{w\in \FuncSty{Light-Children}(u_i)\mid wc(w) = j\}$
            sorted by subtree size}{
                $b_j \leftarrow b_j + \FuncSty{Assign-Weight}(v)$\;
            }
            $a_j(u_i) = \FuncSty{Approx}(a_{j-1}(u_i) + b_j,
            \gamma(u_i)^{-3})$\label{line:approx}\;
        }
        $lw(u_i) = 1 + a_{wc(u_i)}(u_i)$\;
    }
    $k(u_1) = \gamma(u_1) - \ceil{2\lg\gamma(u_1)}+1$\;
    \For{$i=2\to t$}{
        $k(u_i) = \max(\gamma(u_i) - \ceil{2\lg\gamma(u_i)}+1, k(u_{i-1}) - 1)$\;
    }
    \For{$i=t-1\to 1$}{
        $k(u_i) = \max(k(u_i), k(u_{i+1})-1)$\;
    }
    $pw(u_1)\leftarrow 0$\;
    \For{$i=1\to t$}{
        $pw(u_1)\leftarrow pw(u_1) + lw(u_i)+ 2^{k(u_i)} - 1$\;
    }
    \Return $pw(u_1)$\;
\end{algorithm}

The main technical part of this paper is to show that calling
$\FuncSty{Assign-Weight}$ ensures that $pw(u) = O(|T_u|)$ for all apex
nodes, $u\in T$. This is used to show that the maximum id assigned by our
labeling scheme is $O(n)$ and thus takes $\lg n + O(1)$ bits to store.
Intuitively this is the case since the quality of the approximation used in a
node $u$ improves as the size of $u$'s subtree increases. Specifically, we will
use the following lemma, which is proved in Section~\ref{sec:weight_proof}.

\begin{lemma}
    \label{pwlemma}
    Let $T$ be a tree rooted in $r$ and let $u\in T$ be any apex node with
    light height $x$. After calling $\FuncSty{Assign-Weight}(r)$ it holds that:
    \[
        pw(u) \le 3\abs{T_u} \cdot \prod_{i=1}^{x} \left(1 + \frac{6}{i^2}\right)
    \]
    Furthermore, for any node $v\in T$ it holds that
    \begin{align}
    lw(v) \le 3\abs{T^\ell_{v}} \prod_{j=1}^{z} \left(1 +
    \frac{6}{j^2}\right) \cdot
    \left ( 1 + \frac{2}{(z+1)^2} \right )\ ,
    \end{align}
    where $z$ is the maximum light height of any light child of $v$.
\end{lemma}
\begin{corollary}
    \label{pwandlwlinear}
    Let $T$ be a tree rooted in $r$ and let $u\in T$ be any apex node and
    $v\in T$ be any node. After calling $\FuncSty{Assign-Weight}(r)$ it holds that:
    \[
        pw(u) \le 3 e^{\pi^2} \abs{T_u}, \quad
        lw(v) \le 3 e^{\pi^2} \abs{T_v^\ell}
    \]
\end{corollary}
\begin{proof}
Let $u$ be an apex node with light height $x$. Then:
\begin{align*}
    pw(u) &\le 3\abs{T_u} \cdot \prod_{i=1}^{x} \left(1 + \frac{6}{i^2}\right) \\
    &\le 3 \abs{T_u} \cdot \exp \left ( \sum_{i=1}^x \frac{6}{i^2} \right ) \\
    &\le 3 \abs{T_u} \cdot \exp \left ( \sum_{i=1}^\infty \frac{6}{i^2} \right
    ) \\
    &= 3 e^{\pi^2} \abs{T_u}
\end{align*}
The proof for $lw(v)$ is similar.
\end{proof}

\subsection{Id assignment}

We create a procedure $\FuncSty{Assign-Id}(u, s)$ and use it to assign ids to
the nodes in the tree. The procedure takes two parameters: $u$, the node to
which we want to assign the id, and $s$, a lower bound on the id to be assigned.
The function ensures that $id(u) \in \left [s, s+2^{k(u)}-1 \right ]$ has
at least $k(u)$ trailing zero bits and also
assigns an id to every node in $u$'s subtree recursively. We assign ids to
every node in the tree by calling $\FuncSty{Assign-Id}(r,0)$, where $r$ is the
root of the tree. The procedure goes as follows:

\begin{enumerate}
    \item We let $id(u)$ be the unique integer in $\left [s, s+2^{k(u)}-1 \right ]$
          which has at least $k(u)$ trailing zeros in its binary representation.
    \item We let $C_1,\ldots,C_{wc(u)}$ denote the partition of $u$'s light children
          such that every child $v$ with weight class $wc(v) = i$ is contained
          in $C_i$.
    \item Fix $i$ in increasing order. We assign the ids to the nodes in $C_i$ in the following manner.
          For convenience say that $C_i = \set{v_1,\ldots,v_{\abs{C_i}}}$. We then
          let $t_1 = id(u) + a_{i-1}(u) + 1$. For each $j = 1,\ldots,\abs{C_i}$ we
          call $\FuncSty{Assign-Id}(v_j,t_j)$ and set $t_{j+1} = t_j + pw(v_j)$.
    \item Lastly, for the heavy child $v$ of $u$ we call
        $\FuncSty{Assign-Id}(v,id(u)+lw(u))$.
\end{enumerate}

By the above definition we see that for any node $u$ and any node $v\in
T_u^\ell$ we have $id(v)\in (id(u) + a_{wc(v)-1}(u), id(u)+a_{wc(v)}(u)]$. We also
have that $id(u) = id(v) \Leftrightarrow u = v$. Finally, for any two intervals
$I(u),I(v)$ either one is contained in the other or they are disjoint.

Pseudocode for the procedure $\FuncSty{Assign-Id}$ can be found in
Algorithm~\ref{alg:assign_id}.

\begin{algorithm}
    \caption{\FuncSty{Assign-Id}}
    \label{alg:assign_id}
    \DontPrintSemicolon
    \SetKwInOut{Input}{input}\SetKwInOut{Output}{output}
    \Input{Node $u$, First available id $s$.}
    \BlankLine
    $id(u) \leftarrow $ unique integer in $[s, s + 2^{k(u)} - 1]$ with at least
    $k(u)$ trailing zeroes in binary representation.\label{line:skip_0s}\;
    \For{$j = 1\to wc(u)$}{
        $t \leftarrow id(u) + a_{j-1}(u) + 1$\;
        \For{$v\in \{w\in \FuncSty{Light-Children}(u)\mid wc(w) = j\}$
        sorted by subtree size}{
            $\FuncSty{Assign-Id}(v, t)$\;
            $t \leftarrow t + pw(v)$\;
        }
    }
    $\FuncSty{Assign-Id}(heavy(u), id(u) + lw(u))$\;
\end{algorithm}

\subsection{Encoding of labels}
We are now ready to describe the actual labels. Let $u$ be a node. Let
$apex(u) \in \set{0,1}$ and $leaf(u) \in \set{0,1}$ be $1$ if $u$ is an apex
node and a leaf respectively. If $u$ is not a leaf, let $v$ be the heavy child
of $u$ and let $next(u) \in \set{-1,0,1}$ be such that $k(v) = k(u) + next(u)$.
If $u$ is a leaf let $next(u) = 0$. We identify $next(u)$ with the bit string
of size two that is (00) if $next(u) = 0$, (01) if $next(u) = 1$, and (11) if
$next(u) = -1$. We let $aux(u)$ denote the following bit string:
\[
    aux(u) =
    [k(u)]_\gamma \circ
    [wc(u)]_\gamma \circ
    [rld(u)]_\gamma \circ
    apex(u) \circ
    leaf(u) \circ
    next(u)
\]
For each $i=1,2,\ldots,wc(u)$ let $s_i$ be the bit string corresponding to the
$\left (1+ \frac{1}{(\gamma(u))^3} \right )$-approximation $a_i(u)$ as described in
Lemma~\ref{lem:approximation}.
Let $M = \max_i \abs{s_i}$ be the length of the longest of the bit strings and
let $r_i = 0^{M-\abs{s_i}} \circ s_i$. Then $r_1,\ldots,r_{wc(u)}$ have length
$M$. The table, $table(u)$, from which we can decode any of $a_1(u),\ldots,a_{wc(u)}(u)$
in $O(1)$ time is defined as:
\[
    table(u) =
    [M]_\gamma \circ
    r_1 \circ
    \ldots \circ
    r_{wc(u)}
\]
The label of $u$ is then defined as:
\[
    \ell(u) =
    aux(u) \circ
    table(u) \circ
    wlsb(id(u), k(u))
\]

Figure~\ref{fig:table} illustrates how the interval $I(u)$ is split into a part
for each $i\le wc(u)$. in $table(u)$
\begin{figure}[htbp]
    \centering
    \includegraphics[width=.7\textwidth]{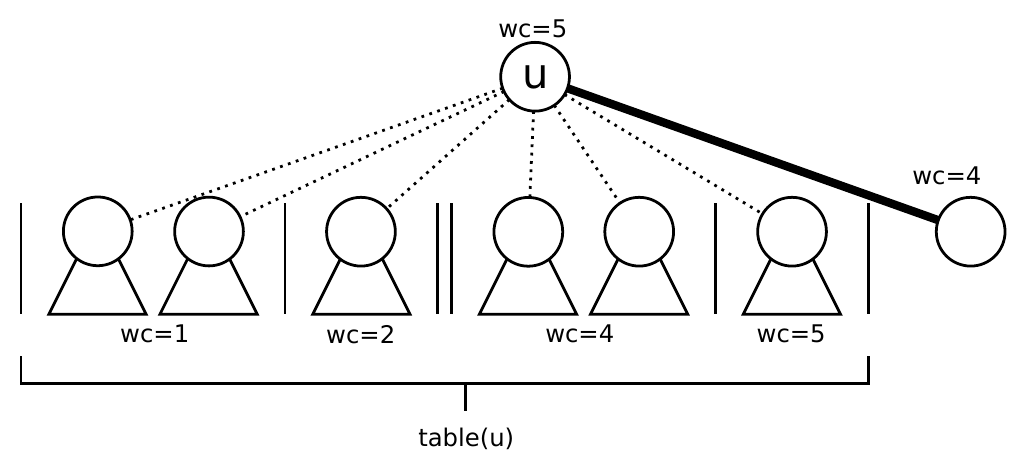}
    \caption{Illustration of the $table(u)$ structure, partitioning $u$'s
    assigned interval into a part for each smaller weight class.}
    \label{fig:table}
\end{figure}

\paragraph{Label size}
Since $rld(u) = O\!\left(\gamma(u)\right)$ by Lemma~\ref{rldBound}
we see that the length of $aux(u)$ is upper bounded by:
\[
    \abs{aux(u)} \le
    2\ceil{\lg k(u)} + O(\lg \gamma(u)) =
    O(\lg k(u))
\]
where we use that $\lg \gamma(u) = O(\lg k(u))$, which is true since
$k(u) \ge k'(u) = \gamma(u) - 2\ceil{\lg \gamma(u)}+1$.

By Corollary~\ref{pwandlwlinear} $lw(u) = O\!\left(\abs{T_u^\ell}\right)$ and hence for every $i=1,\ldots,wc(u)$:
$\lg \lg a_i(u) \le \lg \gamma(u) + O(1)$.
By Lemma~\ref{lem:approximation} we see that $M = O(\lg \gamma(u))$ where $M$ is the
variable used to define $table(u)$. Hence, the length of $table(u)$ is at
most $O\!\left((\lg \gamma(u))^2\right) = O\!\left((\lg k(u))^2\right)$. Furthermore,
the length of $wlsb(id(u),k(u))$ is at most $\ceil{\lg id(u)} - k(u) \le
\lg n - k(u) + O(1)$.
Summarizing, the total label size is upper bounded by:
\[
    \abs{\ell(u)} \le
    \lg n - k(u) + O\!\left((\lg k(u))^2\right) \le
    \lg n + O(1)
\]

\subsection{Decoding}

We will now see how we from two labels $\ell(u), \ell(v)$ of nodes $u,v \in T$ can
deduce whether $u$ is adjacent to $v$. Lemma~\ref{checkParent} below contain necessary and
sufficient conditions for whether $u$ is a parent of $v$.
\begin{lemma}
    \label{checkParent}
    Given two nodes $u,v$: $u$ is a parent of $v$ if and only if either:
    \begin{enumerate}[label=1.\arabic*]
        \item $v$ is a heavy child (i.e. not an apex node).
        \item $u$ is not a leaf.
        \item $id(v)$ is the first number greater than $id(u) + lw(u)$ with at
            least $k(u) + next(u)$ trailing zeroes in its binary
            representation.
    \end{enumerate}
    \vspace{-5pt}
    or:
    \begin{enumerate}[label=2.\arabic*]
        \item $v$ is an apex node.
        \item $wc(v)\le wc(u)$.
        \item $id(v) \in (id(u) + a_{wc(v)-1}(u), id(u) + a_{wc(v)}(u)]$.
        \item If $wc(v) < wc(u)$ then $rld(v) = 0$ else (if $wc(v) = wc(u)$)
            then $rld(v) = rld(u) + 1$.
    \end{enumerate}

%
\end{lemma}
\begin{proof}
    First we will prove that if $v$ is a child of $u$ then either
    1.1, 1.2, 1.3 or 2.1, 2.2, 2.3, 2.4 hold. If $v$ is the heavy
    child of $u$ then clearly 1.1 and 1.2 hold. By definition $id(v)$
    is the unique number in $[id(u)+lw(u),id(u)+lw(u)+2^{k(v)}-1]$ with at
    least $k(v) = k(u) + next(u)$ trailing zeros in its binary representation
    and therefore 1.3 holds.

    Now assume that $v$ is an apex node, i.e. that 2.1 holds. Then $v$ is
    contained in $u$'s light subtree and hence, by definition, 2.2 is
    true. By the definition of \texttt{assign-id} 2.3 holds. 2.4
    follows from Lemma~\ref{rldToParent}.

    Now we will prove the converse. First assume that 1.1, 1.2, 1.3 hold.
    By 1.2, $u$ has a heavy child, $v'$. Since $k(v') = k(u) + next(u)$ we
    see that by 1.3 $id(v') = id(v)$ and hence $v = v'$ and $v$ is a child
    of $u$.

    Now assume that 2.1, 2.2, 2.3, 2.4 hold. By 2.2 and 2.3 we
    know that $v$ is contained in the light subtree of $u$. Assume for the sake
    of contradiction that $v$ is not a child of $u$ and let $v'$ be the child
    of $u$ on the path from $v$ to $u$. By 2.3 we know that $wc(v) =
    wc(v')$. Since there must by at least one light edge on the path from $v$
    to $v'$ (recall that both $v$ and $v'$ are apex nodes)
    Lemma~\ref{rldToApexAncestor} gives that $rld(v') < rld(v)$. But then 2.4
    cannot be true. Contradiction. Hence the assumption was wrong and $v$ is a
    child of $u$.
\end{proof}
In order to check if $u$ is the parent of $v$ we use Lemma~\ref{checkParent}. For $v$
we need to decode:
\[
    apex(v),\ id(v),\ wc(v),\ rld(v)
\]
And for $u$ we need to decode:
\[
    leaf(u),\ wc(u),\ id(u),\ lw(u),\ k(u),\ next(u),\ a_{wc(v)-1}(u),\ a_{wc(v)}(u),\ rld(u)
\]
By the construction of the labels we can clearly do this in $O(1)$ time.

\section{Proof of weight bound}\label{sec:weight_proof}
Below follows the proof of Lemma~\ref{pwlemma}. This is the main technical proof in
this paper.

\begin{proof}[of Lemma~\ref{pwlemma}]
We prove the lemma by induction on $x$. First we prove the lemma when $x = 0$.
Consider a heavy path $p = (u_1,\ldots,u_{\abs{p}})$ in order, where $u_1$ is
closest to the root and has light height $x = 0$. Then $lw(u_i) = 1$ for all
$i = 1,\ldots,\abs{p}$ and:
\[
    pw(u_1) = \sum_{i=1}^{\abs{p}} \left ( lw(u_i) + 2^{k(u_i)} - 1 \right ) =
    \abs{p} + \sum_{i=1}^{\abs{p}} 2^{k(u_i)} - 1 =
    \abs{T_u} + \sum_{i=1}^{\abs{p}} 2^{k(u_i)} - 1
\]
Since $k'(u_i) = 0$ for $i=2,\ldots,\abs{p}$ we see that
$k(u_i) = \max \set{k'(u_1)+1-i,0}$ for any $i$. Hence:
\[
    \sum_{i=1}^{\abs{p}} \left ( 2^{k(u_i)}-1 \right ) \le
    \sum_{i=1}^{\abs{p}} 2^{k'(u_1)+1-i} \le
    \sum_{i=1}^{\infty} 2^{k'(u_1)+1-i} =
    2^{k'(u_1)+1} \le 2\abs{T_u}
\]
Hence $pw(u_1) \le 3\abs{T_u}$ which proves the lemma for $x = 0$.

Assume that the lemma holds for all nodes with light height $< x$, and
consider a heavy path $p = (u_1,\ldots,u_{\abs{p}})$ in order, where $u_1$ has
light height $x$ and is the apex node on $p$. We wish to prove that the
lemma holds for $u_1$.
For each $i = 1,\ldots,\abs{p}$
let $z_i$ be the maximum light-height of any light child of $u_i$.
Let $\alpha(u_i)$ be the sum of $pw(v)$ over all light children $v$ of $u_i$.
For any $i$ we note that $z_i \le x-1$ and so by the induction hypothesis
\[
	\alpha(u_i) \le
	3\left ( \abs{T_{u_i}^\ell} - 1 \right ) \cdot \prod_{j=1}^{z_i} \left(1 + \frac{6}{j^2}\right)
\]
We can upper bound $lw(u_i)$ in terms of $\alpha(u_i)$ by noting that we
approximate the path weights of $u_i$'s children at most $wc(u_i)$ times:
\[
    lw(u_i) =
    1 + a_{wc(u_i)} \le
    1 + \alpha(u_i) \cdot \left ( 1 + \frac{1}{(\gamma(u_i))^3} \right )^{wc(u_i)}
\]
Since $u_i$ has a child with light height $z_i$ it must have a child with a
subtree consisting of at least $2^{z_i+1}-1$ nodes by Lemma~\ref{lightHeightLemma}.
Therefore $\gamma(u_i) \ge z_i+1$.
Since $wc(u_i) = \floor{\lg \gamma(u_i)}$ we can conclude that
\[
    \left ( 1 + \frac{1}{(\gamma(u_i))^3} \right )^{wc(u_i)}
    \le
    1 + \frac{2^{wc(u_i)}}{(\gamma(u_i))^3}
    \le
    1 + \frac{2}{(\gamma(u_i))^2}
    \le
    1 + \frac{2}{(z_i+1)^2}
\]
Combining these observations gives:
\begin{align}
    \label{eq:lwbound}
    lw(u_i) \le
    3\abs{T^\ell_{u_i}} \prod_{j=1}^{z_i} \left(1 + \frac{6}{j^2}\right) \cdot
    \left ( 1 + \frac{2}{(z_i+1)^2} \right )
\end{align}
By an analysis analogous to the one in Section~\ref{sec:caterpillar} we see that:
\begin{align}
    \label{eq:kSumPowerBound}
    \sum_{i=1}^{\abs{p}} 2^{k(u_i)}-1 \le
    3 \sum_{i=1}^{\abs{p}} 2^{k'(u_i)} \le
    6 \sum_{i=1}^{\abs{p}} \frac{2^{\gamma(u_i)}}{(\gamma(u_i))^2}
\end{align}
For any $i=2,\ldots,\abs{p}$ we know that $\gamma(u_i) \ge z_i+1$ and
$2^{\gamma(u_i)} \le \abs{T_{u_i}^\ell}$. Therefore:
\[
    \sum_{i=2}^{\abs{p}} \frac{2^{\gamma(u_i)}}{(\gamma(u_i))^2}
    \le
    \sum_{i=2}^{\abs{p}} \frac{\abs{T_{u_i}^\ell}}{(z_i+1)^2}
\]
By Lemma~\ref{lightHeightLemma} $\abs{T_{u_1}} \ge 2^{x+1}-1$ and therefore
$\gamma(u_1) \ge x$. Hence
$\frac{2^{\gamma(u_1)}}{(\gamma(u_1))^2} \le \frac{\abs{T_{u_1}}}{x^2}$.
Combining these two observations allows us to conclude that
\begin{align}
    \label{eq:gammaSumPowerBound}
    \sum_{i=1}^{\abs{p}} \frac{2^{\gamma(u_i)}}{(\gamma(u_i))^2} \le
    \frac{\abs{T_{u_1}}}{x^2} +
    \sum_{i=2}^{\abs{p}} \frac{\abs{T_{u_i}^\ell}}{(z_i+1)^2} \le
    2 \sum_{i=1}^{\abs{p}} \frac{\abs{T_{u_i}^\ell}}{(z_i+1)^2}
\end{align}
When establishing the last inequality we use that
$\abs{T_{u_1}} = \sum_{i=1}^{\abs{p}} \abs{T_{u_i}^\ell}$.
Now we see that
\begin{align*}
    pw(u_1) & \le
    \sum_{i=1}^{\abs{p}} 3\abs{T^\ell_{u_i}} \prod_{j=1}^{z_i} \left(1 + \frac{6}{j^2}\right) \cdot
    \left ( 1 + \frac{2}{(z_i+1)^2} \right ) +
    \sum_{i=1}^{\abs{p}} \abs{T_{u_i}^\ell} \cdot \frac{12}{(z_i+1)^2} \\
    & \le
    \sum_{i=1}^{\abs{p}} 3\abs{T^\ell_{u_i}} \prod_{j=1}^{z_i} \left(1 + \frac{6}{j^2}\right) \cdot
    \left ( 1 + \frac{6}{(z_i+1)^2} \right ) \\
    &\le 3 \abs{T_{u_1}} \prod_{j=1}^{x} \left(1 + \frac{6}{j^2}\right)
\end{align*}
Here we used
\eqref{eq:lwbound}, \eqref{eq:kSumPowerBound}, and \eqref{eq:gammaSumPowerBound}
together with
the definition of the path weight.
\end{proof}

\section{Running time}\label{sec:time}
In this section we argue that the encoding time of the labeling scheme is
$O(n)$ and the decoding time is $O(1)$, thus finishing the proof of
Theorem~\ref{optimaladjacency}.

\subsection{Encoding time}
To bound the encoding time we will need to bound the total number of nodes with
a given weight class $k$. We will use the following notion of
\emph{contribution}:
\begin{definition}
    For an apex node $u$ we define $contrib(u) = V(T_u)$ and for a heavy child
    $u$ we define $contrib(u) = V(T_u^\ell)$. We say that a node $v\in
    contrib(u)$ is \emph{contributing} to $u$.
\end{definition}
Note that by this definition, the weight class of a node $u$ is exactly
\[
    wc(u) = \floor{\lg \lg |contrib(u)|}\ .
\]
We will need the following lemma:
\begin{lemma}\label{lem:wc_sum}
    Given a tree $T$ with $|T| = n$, the number of nodes $u$ with $wc(u) = k$
    is bounded by
    \[
        O\!\left(n\cdot\frac{2^k}{2^{2^k}}\right)\ .
    \]
\end{lemma}
\begin{proof}
    Consider any node $u\in T$. We will first bound the number of nodes $v$
    with $wc(v) = k$ such that $u\in contrib(v)$. Observe that a node $u$
    contributes to exactly all apex nodes, which are ancestors of $u$ as well
    as the heavy child $v$ of maximum depth for each heavy path $p$, such that
    $v$ is an ancestor of $u$ (note that such $v$ might not exist for a heavy
    path $p$). Thus at least half the nodes that $u$ contributes to are apex nodes.

    Let $w_1$ be the apex node in $T$ of minimum depth such that $w_1$ is an
    ancestor of $u$ and $wc(w_1) = k$. Then $|contrib(w_1)| < 2^{2^{k+1}}$. Let
    $w_i$ be the first apex node on the path from $w_{i-1}$ to $u$ (excluding
    $w_{i-1}$ itself). Then for all $i$ such that $w_i$ is well defined we have
    \[
        |contrib(w_i)| \le |contrib(w_{i-1})|/2\ ,
    \]
    and thus $|contrib(w_{2^k})| < 2^{2^k}$ implying that $wc(w_{2^k}) < k$.
    Thus $u$ can contribute to at most $2^{k+1} + 1$ nodes with weight class
    $k$.

    It follows that the total number of nodes contributing to nodes of weight
    class $k$ is bounded by $n\cdot (2^{k+1}+1)$. Since each node of weight
    class $k$ has at least $2^{2^k}$ nodes contributing to it, we can bound the
    total number of nodes with weight class $k$ by
    \[
        n\cdot\frac{2^{k+1}+1}{2^{2^k}} =
        O\!\left(n\cdot\frac{2^k}{2^{2^k}}\right)\ .
    \]
\end{proof}
The proof of Lemma~\ref{lem:wc_sum} is illustrated in Figure~\ref{fig:num_wc}. The figure
illustrates how each node $u$ contributes to all apex nodes on the path from
$u$ to the root, and how the number of contributing nodes doubles per apex node
on this path.
\begin{figure}[htbp]
    \centering
    \includegraphics[width=.6\textwidth]{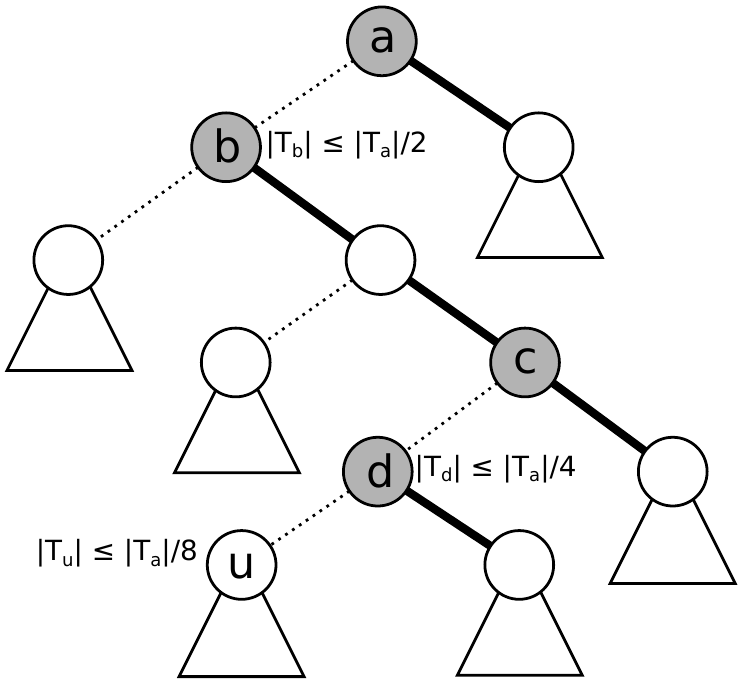}
    \caption{Illustration of Lemma~\ref{lem:wc_sum}. The grey nodes are the ones
        that $u$ are contributing to. For each grey apex node on the path from
        $u$ to the root, the number of contributing nodes grows by at least a
    factor of $2$.}
    \label{fig:num_wc}
\end{figure}

We are now ready to bound the encoding time.
First recall that we are using the word-RAM model with word size $c\log
n$ for some sufficiently large constant $c$ such that the entire label
$\ell(u)$ fits in one word. We are thus able to create the Elias $\gamma$ code
of $k(u)$, $wc(u)$, $rld(u)$, and $M(u)$ in $O(1)$ time for each node $u$ using
standard word operations.

We may assume that the children of each node is sorted by subtree size.
Otherwise we can ensure this using e.g. bucket sort in $O(n)$ time.

Since all components of $aux(u)$ other than $k(u)$ can be calculated using
a simple DFS-traversal in $O(n)$ time, we see that the total encoding time is
dominated by the running time of Algorithm~\ref{alg:assign_weight},
Algorithm~\ref{alg:assign_id}, and the time to construct $table(u)$ from the
$a_i(u)$s. For Algorithm~\ref{alg:assign_id} we first observe that
line~\ref{line:skip_0s} can be done in $O(1)$ time using the following
approach:
\begin{enumerate}
    \item Let $a$ be the integer resulting from setting the last $k(u)$
        bits of the binary representation of $s$ to $0$.
    \item If $a = s$, then return $s$.
    \item Otherwise return $a + 2^{k(u)}$
\end{enumerate}
Each of the three steps can be done in $O(1)$ time using word operations.
The rest of Algorithm~\ref{alg:assign_id} is a DFS-traversal, which runs in $O(n)$
time total. For the construction of $table(u)$, observe that all
of $table(u)$ fits in a word, so we can calculate each $r_i(u)$ in
$O(1)$ time. The total construction time over all nodes of $T$ is thus
bounded by:
\begin{align}\label{eq:bound_wc_sum}
    \begin{split}
    \sum_{u\in T} O(wc(u))
    &= O\!\left(\sum_{k = 0}^{\ceil{\lg \lg n}} k\cdot|\{w\in T\mid wc(w) =
    k\}|\right) \\
    &\le O\!\left(\sum_{k=0}^{\ceil{\lg\lg n}} kn\cdot \frac{2^k}{2^{2^k}}
    \right) \\
    &\le O\!\left(n\cdot\sum_{k=0}^\infty \frac{k\cdot2^k}{2^{2^k}}\right) \\
    &= O(n)\ .
    \end{split}
\end{align}
Here, the second line follows by Lemma~\ref{lem:wc_sum}. For
Algorithm~\ref{alg:assign_weight} we see that the total time spent in the loop
of line~\ref{line:for_weight} to line~\ref{line:approx} for all nodes $u\in T$
is bounded by
\[
    \sum_{u\in T} O(|children(u)| + wc(u))\ .
\]
By \eqref{eq:bound_wc_sum} this is $O(n)$. The rest of
Algorithm~\ref{alg:assign_weight} spends time proportional to the length of the
heavy path the function has been called with, which sums to $O(n)$ over all
heavy paths. Note that line~\ref{line:approx} is calculated in $O(1)$ time
using Lemma~\ref{lem:approximation}.

By summing up the three different parts we see that the total encoding time
of the labeling scheme is $O(n)$.

\subsection{Decoding time}
Using the conditions of Lemma~\ref{checkParent} we will bound the decoding time
of the labeling scheme:

Recall that we are able to decode each of $k(u)$,
$wc(u)$, $rld(u)$, $apex(u)$, $leaf(u)$, $next(u)$, and $M(u)$ in $O(1)$
time. Doing this we also locate the beginning of $a_1(u)$ in the bit string
(label). Let this bit position be denoted by $x$.

Knowing $x$, $M(u)$, and $wc(v)$ we can read the $wc(v)-1$st and
$wc(v)$th entries of $table(u)$ in $O(1)$ time, since these are located
exactly at bit positions $x + M(u)\cdot (wc(v)-2)$ and $x + M(u)\cdot
(wc(v) - 1)$. If $wc(v) = 1$ we know that $a_0(u) = 0$. Similarly we know
that $wlsb(id(u),k(u))$ begins at bit position $x + M(u)\cdot (wc(u) - 1)$ and
consists of the remaining bits. We can do the same for $v$, thus decoding
each relevant component of $\ell(u)$ and $\ell(v)$ can be done in $O(1)$
time.

The conditions 1.1, 1.2 and 2.1-4 can now be checked in $O(1)$ by using the
corresponding values. For condition 1.3 we need to be able to find the
smallest integer greater than $id(u) + lw(u)$ with at least $k(u) +
next(u)$ trailing zeroes. Observe that $lw(u) = 1 + a_{wc(u)}(u)$ can be
obtained in $O(1)$ time from $table(u)$ in the same manner as
$a_{wc(v)}(u)$ was. Finding the smallest such integer can now be done in
$O(1)$ time be using the same procedure as in the previous section.

This finishes the proof of Theorem~\ref{optimaladjacency}.

\newpage

\bibliographystyle{plain}
\bibliography{treeadjacency}

\makeatletter
\def\runninghead{\hrulefill\quad APPENDIX\quad\hrulefill}
\def\ps@headings{
	\def\@oddhead{\footnotesize\rm\hfill\runninghead\hfill}}
\def\@evenhead{\@oddhead}
\def\@oddfoot{\rm\hfill\thepage\hfill}\def\@evenfoot{\@oddfoot}
\makeatother

\newpage
\setlength{\headsep}{15pt} \pagestyle{headings}

\appendix

\section{Adjacency labeling for trees explicitly listed as an open problem} \label{whatisthat}

\newcommand{\q}[1]{``#1''}

Let $T_n$ denote the family of trees with $n$ nodes. In the quotes below ``universal graph'' is ``induced universal''.

Chung~\cite[emphasized on page 452-453]{Chung90}
\q{What is the correct order of magnitude for $g_v(T_n)$? [...] It
would be of particular interest to sharpen the bounds for $g_v(T_n)$ [...]}

In~\cite[page 465]{fraigniaudkorman2}
\q{Proving or disproving the existence of a universal graph with a linear number
of nodes for the class of $n$-node trees is a central open problem in the design
of informative labeling schemes.}

In~\cite[page 592]{gavoille2007shorter}
\q{[...] prove an optimal bound for trees (up to an additive constant) which is still open.}

In~\cite[page 143-144]{bonichon2006short}
\q{leaving open the question of whether trees enjoy a labeling scheme with
$\log n+\Oh(1)$ bit labels [...] In particular, for adjacency queries in trees,
the current lower bound is $\log n$ and the upper bound is $\log n + \Oh(\log^* n)$}

In~\cite[page 42]{cyriltalk}
\q{Induced-universal graph for n-node trees of $\Oh(n)$ size?}

In~\cite{KNR92}
\q{The question of matching upper and lower bounds for the sizes of the universal
graphs for these families still remain open.}
In this paper trees and graphs with bounded arboricity are two of the main families being considered.

In~\cite[page 132]{Fraigniaud2009randomized} \q{Proving or disproving
the existence of an adjacency labeling scheme for trees
using labels of size $\log n + O(1)$ remains a central open problem
in the design of informative labeling schemes.}

\end{document}